\documentclass[12pt]{article}
\usepackage[a4paper, total={7in, 10in}]{geometry}
\usepackage{graphicx} 
\usepackage{amssymb,amsthm,latexsym,amsmath,epsfig,pgf,url}
\newtheorem{theorem}{Theorem}[section]
\newtheorem{lemma}[theorem]{Lemma}
\newtheorem{corollary}[theorem]{Corollary}
\newtheorem{example}[theorem]{Example}
\newtheorem{definition}{Definition}[section]
\newtheorem{remark}{Remark}
\newcommand{\Mod}[1]{\ (\mathrm{mod}\ #1)}
\usepackage{lineno}

\usepackage{authblk}

\title{Self-dual double cyclic codes over $\mathbb{F}_q$}
\author{Ricky Aditya$^{1}$, Aleams Barra$^{2}$\thanks{Corresponding author: Aleams Barra (aleamsbarra@itb.ac.id)}, Djoko Suprijanto$^{3}$\\
	\small $^{1}$Doctoral Program in Mathematics, Faculty of Mathematics and Natural Sciences,\\ Institut Teknologi Bandung, Indonesia 
	\\ $^{2}$Algebra Research Group, Faculty of Mathematics and Natural Sciences,\\ Institut Teknologi Bandung, Indonesia 
	\\
	$^{3}$Combinatorial Mathematics Research Group, Faculty of Mathematics and Natural Sciences,\\ Institut Teknologi Bandung, Indonesia 
}

\date{December 2025}

\begin{document}
	\maketitle
	\setcounter{MaxMatrixCols}{20}
	
	\abstract{This article focuses specifically on the study of self-dual double cyclic codes over a finite field $\mathbb{F}_q$. A self-dual double cyclic code is a double cyclic code that is equal to its dual. Structurally, a double cyclic code of length $(r,s)$ over $\mathbb{F}_q$ is a $\mathbb{F}_q[x]$-submodule of $\mathbb{F}_{q,r,s}:=\mathbb{F}_q[x]/\langle x^r-1\rangle\times\mathbb{F}_q[x]/\langle x^s-1\rangle$. Moreover, any double cyclic code of length $(r,s)$ over $\mathbb{F}_q$ is generated by two pairs of polynomials in $\mathbb{F}_{q,r,s}$. From the properties of the generating elements, we provide the necessary and sufficient conditions such that two pairs of polynomials in $\mathbb{F}_{q,r,s}$ generate a self-dual code. Furthermore, we examine the existence of self-dual double cyclic codes for some specific lengths: $(r,r)$; $(r,2r)$ and $(2r,r)$; and $(r,s)$, where $\gcd(r,s)=1$. For each case, we provide a construction method with some explicit examples over various finite fields. We also observe some connections between self-dual double cyclic codes and other classes of self-dual codes.}
	
	
	\section{Introduction}
	In classical coding theory, cyclic codes are a special class of linear codes which have a strong algebraic structure that enables an efficient encoding and decoding process. The efficiency of cyclic codes makes them suitable for many practical applications that require high-speed data processing. The structure of a cyclic code of length $n$ over a finite field $\mathbb{F}_q$ can be identified with an ideal of the polynomial ring $\mathbb{F}_q[x]/\langle x^n-1\rangle$. Hence, the existence of cyclic codes of length $n$ over $\mathbb{F}_q$ depends on the factorization of $x^n-1$ over $\mathbb{F}_q$.
	
	Since the number of cyclic codes of a specific length is quite limited, a considerable amount of research on their generalizations has been carried out. Some of them are quasi-cyclic codes, negacyclic codes, constacyclic codes, skew-cyclic codes, and double cyclic codes. The goal is to obtain a wider variety of codes with good parameters while preserving the main structural properties. Some classes of optimal codes can be constructed using this approach.
	
	In this article, we focus on double cyclic codes. The concept of double cyclic codes was introduced by Borges et al. \cite{Borges2018}. In their article, they worked specifically on double cyclic codes over $\mathbb{Z}_2$. However, some of their results can easily be generalized to codes over any finite field $\mathbb{F}_q$. Later, Gao et al. \cite{Gao2016} worked on double cyclic codes over $\mathbb{Z}_4.$  Then several other works on double cyclic codes over various finite fields or finite rings have been done, such as over the finite field $\mathbb{F}_p$ \cite{Abualrub2024}, the finite ring $\mathbb{Z}_2+u\mathbb{Z}_2$, where $u^2=0$ \cite{Aydogdu2024}, and the finite ring $\mathbb{F}_q+v\mathbb{F}_q$, where $v^2=v$ \cite{Deng2020}. Furthermore, the concept of double skew-cyclic codes has also been developed \cite{Aydogdu2022}.
	
	On the other hand, self-dual codes are another special class of linear codes which have a unique and rich mathematical structure. They have deep connections to other mathematical fields that enable the construction of codes with desirable properties. There are several studies that focus on the self-dual case of a special class of codes, such as self-dual cyclic codes \cite{Jia2011,Zhang2024}, self-dual quasi-cyclic codes \cite{Ling2001,Choi2024}, self-dual negacyclic codes \cite{Guenda2012,Jitman2020}, and self-dual constacyclic codes \cite{Yang2015}.
	
	Therefore, the concept of self-dual double cyclic codes is interesting to explore. In this article, we discuss some aspects of self-dual double cyclic codes over a finite field $\mathbb{F}_q$, including their existence and their construction methods. We also explore some connections between self-dual double cyclic codes and other classes of self-dual codes, such as self-dual cyclic codes, self-dual quasi-cyclic codes, and self-dual negacyclic codes.
	
	This article is organized as follows. In Section 2 we recall some basic facts regarding the double cyclic codes over $\mathbb{F}_q$ and their dual codes. 
	The existence of self-dual double cyclic codes over $\mathbb{F}_q$ and their characterization is observed in Section 3. The 
	last three sections discuss self-dual double cyclic codes of a specific length. In Section 4 we examine a special case of self-dual double cyclic codes of length $(r,r)$ over $\mathbb{F}_q$. In Section 5, another special case of self-dual double cyclic codes of length $(r,2r)$ and $(2r,r)$ over $\mathbb{F}_q$ is examined. Then Section 6 is devoted to 
	self-dual double cyclic codes of length $(r,s)$ over $\mathbb{F}_q$, where $\gcd(r,s)=1$. The article ends with some concluding remarks.
	
	\section{Double cyclic codes over $\mathbb{F}_q$}
	Double cyclic codes are a generalization of classical cyclic codes. In \cite{Borges2018}, Borges et al. discussed the structure and properties of double cyclic codes over $\mathbb{Z}_2$, which can be generalized to double cyclic codes over any finite field $\mathbb{F}_q$. The definitions and theorems in this section are based on \cite{Borges2018}, with a generalization from codes over $\mathbb{Z}_2$ to codes over $\mathbb{F}_q$.
	
	In double cyclic codes, any codeword of length $n$, where $n=r+s$, is partitioned into two parts: the first $r$ coordinates and the last $s$ coordinates, and each part is invariant under cyclic shift. For the remainder of this article, the notation $(r,s)$ is used to indicate how the codewords are partitioned.
	
	\begin{definition}
		Let $C$ be a linear code of length $n=r+s$ over $\mathbb{F}_q$ and let
		$$u=(u_0,u_1,\cdots,u_{r-1}|u'_0,u'_1,\cdots,u'_{s-1})\in C$$
		be a codeword in $C$. The $(r,s)$-double cyclic shift of $u$, denoted by $\sigma_{r,s}(u)$, is defined as
		$$\sigma_{r,s}(u)=(u_{r-1},u_0,\cdots,u_{r-2}|u'_{s-1},u'_0,\cdots,u'_{s-2}).$$
		Moreover, $C$ is a double cyclic code of length $(r,s)$ if it is invariant under $(r,s)$-double cyclic shift, i.e., for every $u\in C$, $\sigma_{r,s}(u)\in C$.
	\end{definition}
	
	Define $\mathbb{F}_{q,r,s}=\mathbb{F}_q [x]/\langle x^r -1\rangle \times \mathbb{F}_q [x]/\langle x^s -1\rangle$. It is clear that $\mathbb{F}_{q,r,s}$ is an $\mathbb{F}_q [x]$-module with respect to pointwise polynomial addition and scalar multiplication defined as
	$$f(x) (p(x),q(x))=(f(x) p(x) \Mod{x^r-1},f(x) q(x) \Mod{x^s-1}),$$
	for every $f(x)\in \mathbb{F}_q [x]$ and for every $(p(x),q(x))\in\mathbb{F}_{q,r,s}$. If the codeword $$u=(u_0,u_1,\cdots,u_{r-1}|u'_0,u'_1,\cdots,u'_{s-1})$$ is associated with an ordered pair of polynomials
	$$(u(x),u'(x))=(u_0+u_1 x+\cdots+u_{r-1} x^{r-1}, u'_0+u'_1 x+\cdots+u'_{s-1} x^{s-1})\in \mathbb{F}_{q,r,s},$$
	then $\sigma_{r,s}(u)$ is associated with $$(x u(x) \Mod{x^r-1},x u'(x) \Mod{x^s-1})=x(u(x),u'(x)).$$
	Therefore, structurally a double cyclic code of length $(r,s)$ over $\mathbb{F}_q$ is an $\mathbb{F}_q [x]$-submodule of $\mathbb{F}_{q,r,s}$.
	
	In general, double cyclic codes are not the same as the direct product of two cyclic codes. The canonical projections of a double cyclic code $C$ of length $(r,s)$ over $\mathbb{F}_q$ are defined as
	\begin{align*}
		\pi_r(C)=\{p(x)\in \mathbb{F}_q [x]/\langle x^r -1\rangle:~ (p(x),q(x))\in C\},\\
		\pi_s(C)=\{q(x)\in \mathbb{F}_q [x]/\langle x^s -1\rangle:~ (p(x),q(x))\in C\}.
	\end{align*}
	It can be shown that $\pi_r(C)$ and
	$\pi_s(C)$ are $\mathbb{F}_q[x]$-submodules of $\mathbb{F}_q[x]/\langle x^r-1\rangle$  and $\mathbb{F}_q[x]/\langle x^s-1\rangle$, respectively. This means that both $\pi_r(C)$ and $\pi_s(C)$ are cyclic codes over $\mathbb{F}_q$ of length $r$ and $s$, respectively. However, in general $C$ and $\pi_r(C)\times \pi_s(C)$ are not always equal (see Example \ref{exsddc1} in the next section). A double cyclic code $C$ of length $(r,s)$ over $\mathbb{F}_q$ such that $C=\pi_r(C)\times \pi_s(C)$ is called a separable double cyclic code.
	
	As in classical cyclic codes, double cyclic codes also have generating elements. The following theorem gives important structural properties of the generating elements of double cyclic codes, which is a generalization of Theorem 3.1 in \cite{Borges2018} and Theorem 3 in \cite{Abualrub2024}.
	\begin{theorem}\label{gendouble}
		Every double cyclic code of length $(r,s)$ over $\mathbb{F}_q$ is generated by two elements $(b(x),0),(l(x),a(x))\in\mathbb{F}_{q,r,s}$, where $a(x),b(x),l(x)\in\mathbb{F}_q[x]$ such that $b(x)|(x^r-1)$ and $a(x)|(x^s-1)$.
	\end{theorem}
	
	Theorem \ref{gendouble} above shows that double cyclic codes are finitely generated submodules (to be precise, generated by two elements). Here, $a(x)$ is a generator polynomial for $\pi_s(C)$, $b(x)$ is a generator polynomial for $C'=\{p(x)\in \mathbb{F}_q [x]/\langle x^r -1\rangle:~ (p(x),0)\in C\}$, and $l(x)$ is a polynomial in $\mathbb{F}_q[x]/\langle x^r-1\rangle$ such that $(l(x),a(x))\in C$. Moreover, from Proposition 3.5 in \cite{Borges2018}, we can assume that $\deg(l(x))<\deg(b(x))$, and from Proposition 3.6 in \cite{Borges2018}, we have $b(x)$ divides $\frac{x^s-1}{a(x)}l(x)$. If $C$ is a separable double cyclic code, then $\pi_r(C)=C'$ and $l(x)=0$.
	
	Furthermore, it can be shown that the minimal generating set of a double cyclic code of length $(r,s)$ over $\mathbb{F}_q$ generated by $(b(x),0)$ and $(l(x),a(x))$ is $$\{(b(x),0),x(b(x),0),\cdots,x^{r-\beta-1}(b(x),0),(l(x),a(x)),x(l(x),a(x)),\cdots,x^{s-\alpha-1}(l(x),a(x))\},$$ where $\alpha=\deg(a(x))$ and $\beta=\deg(b(x))$ (see \cite{Abualrub2024}, \cite{Aydogdu2024}, \cite{Borges2018}, \cite{Deng2020}, \cite{Gao2016}). Thus, as a linear code over $\mathbb{F}_q$, its dimension is $r+s-\deg(b(x))-\deg(a(x))$.
	
	Let $C$ be a double cyclic code of length $(r,s)$ over $\mathbb{F}_q$.
	As in linear codes in general, the dual code $C^{\perp}$ is defined as the orthogonal complement of $C$ in $\mathbb{F}_q^n$, where $n=r+s$. Here, for any $x\in C$, elements of $C^{\perp}$ are orthogonal to $x$ and all its $(r,s)$-double cyclic shifts. Since a double cyclic code of length $(r,s)$ is an $\mathbb{F}_q[x]$-submodule of $\mathbb{F}_{q,r,s}$, its orthogonal complement is also an $\mathbb{F}_q[x]$-submodule of $\mathbb{F}_{q,r,s}$. Thus, the dual code of a double cyclic code is also a double cyclic code.
	
	To derive the properties of the dual of double cyclic codes, we need to characterize the orthogonality of two ordered pairs of polynomials in $\mathbb{F}_{q,r,s}$. For any polynomial $f(x)$ of degree $k$ over $\mathbb{F}_q$, the reciprocal polynomial of $f(x)$ is defined as $f^{*}(x)= x^{k} f(x^{-1})$. It is clear that $f^{**}(x)=f(x)$, and $f(x)$ is said to be self-reciprocal if $f^{*}(x)=f(x)$. Let $m=\operatorname{lcm}(r,s)$, define an operation $\circ:\mathbb{F}_{q,r,s}\times \mathbb{F}_{q,r,s} \rightarrow \mathbb{F}_q[x]/\langle x^m -1\rangle$ as:
	\begin{align*}
		(u(x),u'(x))\circ (v(x),v'(x)) & =\left[u(x) v^{*}(x) \Mod{x^r-1}\right]\left(\frac{x^m-1}{x^r-1}\right) \\
		&\quad + \left[u'(x) v'^{*}(x)\Mod{x^s -1}\right]\left(\frac{x^m-1}{x^s-1}\right),
	\end{align*}
	for every $(u(x),u'(x)),(v(x),v'(x))\in \mathbb{F}_{q,r,s}$.
	
	In \cite{Aydogdu2024}, Aydogdu showed that $(u(x),u'(x))\circ (v(x),v'(x))=0$ is equivalent to $(u,u')$ being orthogonal to $(v,v')$ and all of its $(r,s)$-double cyclic shifts. This will be used to determine the generating elements of the dual code of double cyclic codes.
	
	Let a double cyclic code $C$ be generated by $(b(x),0)$ and $(l(x),a(x))$, and the dual code $C^{\perp}$ be generated by $(\bar{b}(x),0)$ and $(\bar{l}(x),\bar{a}(x))$. The following theorem provides the relations among $(b(x),0)$, $(l(x),a(x))$, $(\bar{b}(x),0)$, and $(\bar{l}(x),\bar{a}(x))$, which is a generalization of Proposition 4.12, Proposition 4.13, and Proposition 4.16 in \cite{Borges2018}.
	\begin{theorem}\label{gendual}
		If $C=\langle(b(x),0),(l(x),a(x))\rangle$ is a double cyclic code over $\mathbb{F}_q$ and $C^{\perp}=\langle(\bar{b}(x),0),(\bar{l}(x),\bar{a}(x))\rangle$ is its dual code, then
		\begin{align*}
			\bar{b}^{*}(x)&=\frac{x^r-1}{\gcd(b(x),l(x))}, \quad\bar{l}^{*}(x)=\gamma(x)\frac{x^r-1}{b(x)}, \\
			\text{ and }&\quad\bar{a}^{*}(x)=\frac{\gcd(b(x),l(x))(x^s-1)}{b(x) a(x)},&
		\end{align*}
		for some $\gamma(x)\in\mathbb{F}_q[x]$.
	\end{theorem}
	\begin{proof}
		From the definition of dual codes, each element of $C$ is orthogonal to each element of $C^{\perp}$. Therefore, we have
		\begin{itemize}
			\item $(b(x),0)\circ (\bar{b}(x),0)=0 \Leftrightarrow b(x) \bar{b}^{*}(x) \Mod{x^r-1}=0$,
			\item $(b(x),0)\circ (\bar{l}(x),\bar{a}^{*}(x))=0 \Leftrightarrow b(x) \bar{l}^{*}(x) \Mod{x^r-1}=0$,
			\item $(l(x),a(x))\circ (\bar{b}(x),0)=0 \Leftrightarrow l(x) \bar{b}^{*}(x) \Mod{x^r-1}=0$,
			\item $(l(x),a(x))\circ (\bar{l}(x),\bar{a}^{*}(x))=0 \Leftrightarrow \left[l(x) \bar{l}^{*}(x) \Mod{x^r-1}\right]\left(\frac{x^m-1}{x^r-1}\right)\\
			+ \left[a(x) \bar{a}^{*}(x)\Mod{x^s -1}\right]\left(\frac{x^m-1}{x^s-1}\right)=0$.
		\end{itemize}
		From the first and third conditions, we have that $x^r-1$ divides both $b(x)\bar{b}^{*}(x)$ and $l(x)\bar{b}^{*}(x)$. Therefore, $x^r-1$ divides $\gcd(b(x),l(x))\bar{b}^{*}(x)$. Since $b(x)$ divides $x^r-1$, if $\bar{b}(x)$ is the polynomial of the smallest degree that satisfies both conditions, then $$\gcd(b(x),l(x))\bar{b}^{*}(x)=x^r-1\Leftrightarrow \bar{b}^{*}(x)=\frac{x^r-1}{\gcd(b(x),l(x))}.$$
		Since $(b(x),0)$ and $(l(x),a(x))$ generate $C$, the pair of polynomials $$-\frac{l(x)}{\gcd(b(x),l(x))}(b(x),0)+\frac{b(x)}{\gcd(b(x),l(x))}(l(x),a(x))=\left(0,\frac{b(x)a(x)}{\gcd(b(x),l(x))}\right)$$ are in $C$. Thus, it is orthogonal to $(\bar{l}(x),\bar{a}(x))$, and we obtain $\frac{b(x)a(x)\bar{a}^{*}(x)}{\gcd(b(x),l(x))}\Mod{x^s-1}=0$. Therefore, $x^s-1$ divides $\frac{b(x)a(x)\bar{a}^{*}(x)}{\gcd(b(x),l(x))}$. Since $a(x)$ divides $x^s-1$ and $b(x)$ divides $\frac{x^s-1}{a(x)}l(x)$, we have $\frac{b(x)}{\gcd(b(x),l(x))}$ divides $\frac{x^s-1}{a(x)}\frac{l(x)}{\gcd(b(x),l(x))}$. Since $\frac{b(x)}{\gcd(b(x),l(x))}$ and $\frac{l(x)}{\gcd(b(x),l(x))}$ are relatively prime, $\frac{b(x)}{\gcd(b(x),l(x))}$ must divide $\frac{x^s-1}{a(x)}$. If $\bar{a}(x)$ is the polynomial of the smallest degree that satisfies that condition, then $$\frac{b(x)a(x)\bar{a}^{*}(x)}{\gcd(b(x),l(x))}=x^s-1\Leftrightarrow \bar{a}^{*}(x)=\frac{\gcd(b(x),l(x))(x^s-1)}{b(x)a(x)}.$$
		From the second condition we have $x^r-1$ divides $b(x)\bar{l}^{*}(x)$. Since $b(x)$ divides $x^r-1$, $\frac{x^r-1}{b(x)}$ divides $\bar{l}^{*}(x)$, i.e., $\bar{l}^{*}(x)=\gamma(x)\frac{x^r-1}{b(x)}$, for some $\gamma(x)\in\mathbb{F}_q[x]$. Moreover, the fourth condition becomes
		\begin{align*}
			&\left[l(x) \gamma(x)\frac{x^r-1}{b(x)} \Mod{x^r-1}\right]\left(\frac{x^m-1}{x^r-1}\right)+ \frac{\gcd(b(x),l(x))(x^s-1)}{b(x)}\left(\frac{x^m-1}{x^s-1}\right)=0\\
			\Leftrightarrow &\left[l(x) \gamma(x)\frac{x^r-1}{b(x)} \Mod{x^r-1}\right]\left(\frac{x^m-1}{x^r-1}\right)+ \gcd(b(x),l(x))\frac{x^r-1}{b(x)}\left(\frac{x^m-1}{x^r-1}\right)=0\\
			\Leftrightarrow &\left[(l(x) \gamma(x)+\gcd(b(x),l(x)))\frac{x^r-1}{b(x)} \Mod{x^r-1}\right]\left(\frac{x^m-1}{x^r-1}\right)=0.
		\end{align*}
		The last equation is satisfied if $b(x)$ divides $l(x) \gamma(x)+\gcd(b(x),l(x))$. Note that the existence of $\gamma(x)$ here is guaranteed by B\'ezout's identity.
	\end{proof}
	For separable double cyclic codes, using $l(x)=0$, and consequently, $\gcd(b(x),l(x))=b(x)$, the generating elements can be simplified as a special case of Theorem \ref{gendual}.
	\begin{corollary}\label{gendualsep}
		If $C$ is a separable double cyclic code over $\mathbb{F}_q$ generated by $(b(x),0)$ and $(0,a(x))$, then the dual code $C^{\perp}$ is also a separable double cyclic code over $\mathbb{F}_q$ generated by $(\bar{b}(x),0)$ and $(0,\bar{a}(x))$ such that
		\begin{align*}
			\bar{b}^{*}(x)=\frac{x^r-1}{b(x)} \text{ and } \bar{a}^{*}(x)=\frac{x^s-1}{a(x)}.
		\end{align*}
	\end{corollary}
	
	Corollary \ref{gendualsep} above is the generalization of Proposition 4.14 in \cite{Borges2018}, and it also shows that the dual code of a separable double cyclic code is also separable. Moreover, $\bar{b}(x)$ and $\bar{a}(x)$ are the generator polynomials for the dual code of cyclic codes generated by $b(x)$ and $a(x)$, respectively. From the previous section, we have $\pi_r(C)=\langle b(x)\rangle$ and $\pi_s(C)=\langle a(x)\rangle$.  Therefore, we can conclude that $\pi_r(C^{\perp})=\langle\bar{b}(x)\rangle=\langle b(x)\rangle^{\perp}=(\pi_r(C))^{\perp}$ and $\pi_s(C^{\perp})=\langle\bar{a}(x)\rangle=\langle a(x)\rangle^{\perp}=(\pi_s(C))^{\perp}$. As a consequence, $C^{\perp}=(\pi_r(C))^{\perp}\times (\pi_s(C))^{\perp}$.
	
	\section{Characterization of self-dual double cyclic codes over $\mathbb{F}_q$}
	A linear code $C$ is said to be a self-dual code if it is equal to its dual, i.e., $C^{\perp}=C$. Because of their unique and rich structure, self-dual codes become an interesting area to explore. Their deep connections to other mathematical fields such as unimodular lattices, invariant theory, and combinatorial design theory also enable the construction of codes with desirable properties.
	
	A cyclic code that is also a self-dual code is called a self-dual cyclic code. Jia et al. \cite{Jia2011} worked specifically on self-dual cyclic codes over a finite field $\mathbb{F}_q.$  An important result in their article is the necessary and sufficient conditions for the existence of self-dual cyclic codes.
	\begin{theorem}[\cite{Jia2011}, Section III, Theorem I]\label{Jia31}
		There exists at least one self-dual cyclic code of length $n$ over $\mathbb{F}_q$ if and only if $q$ is a power of 2 and $n$ is even.
	\end{theorem}
	Note that if $q$ is a power of $2,$ then for any even positive integer $n$, the cyclic code generated by $x^{n/2}+1$ is self-dual of length $n$ over $\mathbb{F}_{q}$. This code is called the trivial self-dual cyclic code. In \cite{Jia2011}, it was also shown that for some even positive integer $n$, a nontrivial self-dual cyclic code of length $n$ over $\mathbb{F}_{q}$ exists, but for some other $n$, only the trivial one exists.
	
	In this section, we discuss self-dual double cyclic codes, i.e. double cyclic codes which are also self-dual codes. For the separable case, the following theorem is a simple extension from Theorem \ref{Jia31}.
	\begin{theorem}\label{existsep}
		There exists a separable self-dual double cyclic code of length $(r,s)$ over $\mathbb{F}_q$ if and only if $q$ is a power of 2 and both $r$ and $s$ are even.
	\end{theorem}
	\begin{proof}
		Let $C$ be a separable double cyclic code over $\mathbb{F}_q$. We have $C=\pi_r(C)\times \pi_s(C)$ and $C^{\perp}=(\pi_r(C))^{\perp}\times (\pi_s(C))^{\perp}$. Hence, if $C$ is self-dual, then both $\pi_r(C)$ and $\pi_s(C)$ are self-dual cyclic codes over $\mathbb{F}_q$. From Theorem \ref{Jia31}, both $r$ and $s$ have to be even, and $q$ is a power of 2. Conversely, if $q$ is a power of 2, and both $r$ and $s$ are even, then we have at least the trivial self-dual cyclic codes of length $r$ and $s$ over $\mathbb{F}_q$, that are generated by $x^{r/2}+1$ and $x^{s/2}+1$, respectively. Furthermore, we have at least one separable self-dual cyclic code of length $(r,s)$ over $\mathbb{F}_q$, that is generated by $(x^{\frac{r}{2}}+1,0)$ and $(0,x^{\frac{s}{2}}+1)$.
	\end{proof}
	
	Now we consider the nonseparable case, where the double cyclic code is not equal to the direct product of its canonical projections. First of all, for some $q,r,s\in\mathbb{N}$, nonseparable self-dual double cyclic codes of length $(r,s)$ over $\mathbb{F}_q$ do indeed exist, as shown in the following example.
	
	\begin{example}\label{exsddc1}
		Here $q=2$ and $r=s=4$. Let $C$ be a double cyclic code of length $(4,4)$ over $\mathbb{F}_2$ 
		generated by $((1+x)^3,0)$ and $(1+x,1+x)$. If the codewords are represented as binary vectors of length 8, then $C$ is generated by $$S=\{1111\;0000,1100\;1100,0110\;0110,0011\;0011\}.$$
		It can be shown that $S$ is an orthogonal set over $\mathbb{F}_2$ and
		\begin{align*}
			C=\langle S\rangle=\{&0000\;0000,0000\;1111,0011\;0011,0011\;1100,0101\;0101,0101\;1010,\\
			& 0110\;0110,0110\;1001,1001\;0110,1001\;1001,1010\;0101,1010\;1010,\\
			& 1100\;0011,1100\;1100,1111\;0000,1111\;1111\}.
		\end{align*}
		Thus, $C$ is a self-dual code of length 8, dimension 4, and minimum distance 4 over $\mathbb{F}_2$. This code is an optimal code since the only binary MDS codes are the trivial ones. Moreover, the canonical projections of $C$ are
		$$\pi_r(C)=\{0000,0011,0101,0110,1001,1010,1100,1111\}$$ and
		$$\pi_s(C)=\{0000,0011,0101,0110,1001,1010,1100,1111\}.$$
		Here, $C\neq \pi_r(C)\times\pi_s(C)$, since $0011\in\pi_r(C)$ and $0101\in\pi_s(C)$, but $0011\;0101\notin C$.
	\end{example}
	Note that in Example \ref{exsddc1}, the canonical projections are cyclic codes, but they are not self-dual. Therefore, if $C$ is a self-dual double cyclic code that is not separable, then $\pi_r(C)$ and $\pi_s(C)$ do not need to be self-dual cyclic codes.
	Furthermore, this fact opens up the possibility that nonseparable self-dual double cyclic codes may exist over a finite field of characteristic other than 2. The following theorem gives a characterization of self-dual double cyclic codes related to their generating sets.
	\begin{theorem}\label{charsddc}
		Let $C$ be a double cyclic code of length $(r,s)$ over $\mathbb{F}_q$ generated by $(b(x),0)$ and $(l(x),a(x))$, and let $b(x)=d(x)\beta(x)$ and $l(x)=d(x)\lambda(x)$, where $d(x)=\gcd(b(x),l(x))$. The code $C$ is self-dual if and only if these three conditions are satisfied:
		\begin{itemize}
			\item $d(x) b^{*}(x) = x^r-1$,
			\item $a(x)a^{*}(x)\beta(x)=x^s-1$,
			\item $\beta(x)$ divides $1-\lambda(x)\lambda^{*}(x)$.
		\end{itemize}
	\end{theorem}
	\begin{proof}
		($\Rightarrow$) Let the dual code $C^{\perp}$ be generated by $(\bar{b}(x),0)$ and $(\bar{l}(x),\bar{a}(x))$. If $C$ is self-dual, then $C$ and $C^{\perp}$ have the same generating sets. Using Theorem \ref{gendual}, from $b(x)=\bar{b}(x)$ we have $$b^{*}(x)=\frac{x^r-1}{\gcd(b(x),l(x))}\Leftrightarrow \gcd(b(x),l(x)) b^{*}(x) = d(x) b^{*}(x)= x^r-1,$$ and from $a(x)=\bar{a}(x)$ we have $$
		a^{*}(x)=\frac{\gcd(b(x),l(x))(x^s-1)}{b(x) a(x)} \Leftrightarrow \frac{a(x)a^{*}(x)b(x)}{\gcd(b(x),l(x))}=a(x)a^{*}(x)\beta(x)=x^s-1.$$
		Thus, the first two conditions are satisfied. Moreover, from $l(x)=\bar{l}(x)$ we have $$l^{*}(x)=\gamma(x) \frac{x^r-1}{b(x)}\Leftrightarrow b(x)l^{*}(x)=\gamma(x) (x^r-1).$$ Since $d(x)b^{*}(x)=d(x)d^{*}(x)\beta^{*}(x)=x^r-1$, taking the reciprocal of both sides, we have $d(x)d^{*}(x)\beta(x)=-(x^r-1)$. As a consequence, $$d(x)\beta(x) d^{*}(x)\lambda^{*}(x)=-\gamma(x)d^{*}(x)d(x)\beta(x)\Leftrightarrow \lambda^{*}(x)=-\gamma(x).$$
		Recall that $\gamma(x)$ is a polynomial over $\mathbb{F}_q$ such that $b(x)$ divides $d(x)+\gamma(x)l(x)$. This is equivalent to $\beta(x)$ divides $1+\gamma(x)\lambda(x)=1-\lambda(x)\lambda^{*}(x)$, i.e. the third condition is satisfied.
		
		($\Leftarrow$) First, consider the case where $\beta(x)$ is a nonzero constant polynomial and therefore a unit in $\mathbb{F}_q[x]$. In this case, $b(x)$ divides $d(x)$, and consequently also divides $l(x)=d(x)\lambda(x)$. Since $-\beta(x)^{-1}\lambda(x)(b(x),0)+(l(x),a(x))=(0,a(x))$, it follows that the code generated by $(b(x),0)$ and $(l(x),a(x))$ is equal to the code generated by $(b(x),0)$ and $(0,a(x))$. Thus, it is a separable double cyclic code. From $d(x)b^*(x)=d(x)d^*(x)\beta^*(x)=x^r-1$ and $a(x)a^{*}(x)\beta(x)=x^s-1$, we have $b(x)b^{*}(x)\Mod{x^r-1}=0$ and $a(x)a^{*}(x)\Mod{x^s-1}=0$. Thus, $(b(x),0)$ and $(0,a(x))$ form an orthogonal set and hence the code is self-dual.
		
		Now consider the other case where $\beta(x)$ is not a constant polynomial. Let $m=\operatorname{lcm}(r,s)$. Recall that $\{(b(x),0),(l(x),a(x))\}$ is an orthogonal set if and only if these three conditions are satisfied:
		\begin{itemize}
			\item $b(x)b^{*}(x)\Mod{x^r-1}=0$,
			\item $b(x)l^{*}(x)\Mod{x^r-1}=0$,
			\item $\left[l(x)l^{*}(x)\Mod{x^r-1}\right] \left(\frac{x^m-1}{x^r-1}\right)+\left[a(x)a^{*}(x)\Mod{x^s-1}\right]\left(\frac{x^m-1}{x^s-1}\right)=0$.
		\end{itemize}
		Since $d(x) b^{*}(x) = x^r-1$, we have $b(x)b^{*}(x)=\beta(x) d(x)b^{*}(x)=\beta(x)(x^r-1)$ and $b(x)l^{*}(x)=b(x) d^{*}(x)\lambda^{*}(x)=-\lambda^{*}(x)(x^r-1)$. In other words, $b(x)b^{*}(x)$ and $b(x)l^{*}(x)$ are both divided by $x^r-1$. These satisfy the first two conditions.
		
		From $a(x)a^{*}(x)\beta(x)=x^s-1$ and $\beta(x)$ being a nonconstant polynomial, it is clear that $\deg(a(x)a^{*}(x))<s$. Since $\beta(x)$ divides $1-\lambda(x)\lambda^{*}(x)$, we can write $1-\lambda(x)\lambda^{*}(x)=\beta(x)\nu(x)$, for some $\nu(x)\in\mathbb{F}_q[x]$. Then we obtain
		\begin{align*}
			\left[l(x)l^{*}(x)\Mod{x^r-1}\right]&\left(\frac{x^m-1}{x^r-1}\right)+\left[a(x)a^{*}(x)\Mod{x^s-1}\right]\left(\frac{x^m-1}{x^s-1}\right)\\
			&=\left[l(x)l^{*}(x)\Mod{x^r-1}\right]\left(\frac{x^m-1}{x^r-1}\right)+\frac{x^s-1}{\beta(x)}\frac{x^m-1}{x^s-1}\\
			&=\left[\left(l(x)l^{*}(x)+\frac{x^r-1}{\beta(x)}\right)\Mod{x^r-1}\right]\left(\frac{x^m-1}{x^r-1}\right)\\
			&=\left[d(x)d^{*}(x)(\lambda(x)\lambda^{*}(x)-1)\Mod{x^r-1}\right]\left(\frac{x^m-1}{x^r-1}\right)\\
			&=\left[-d(x)d^{*}(x)\beta(x)\nu(x)\Mod{x^r-1}\right]\left(\frac{x^m-1}{x^r-1}\right)\\
			&=\left[\nu(x)(x^r-1)\Mod{x^r-1}\right]\left(\frac{x^m-1}{x^r-1}\right)=0
		\end{align*}
		This satisfies the third condition. Therefore, $C=\langle (b(x),0),(l(x),a(x))\rangle$ is a self-dual double cyclic code of length $(r,s)$ over $\mathbb{F}_q$.
	\end{proof}
	
	Note that $\lambda(x)=1$ always satisfies the third condition of Theorem \ref{charsddc}. In other words, to construct a self-dual double cyclic code, we only need to find polynomials $d(x)$, $\beta(x)$, and $a(x)$ that satisfy the first two conditions of Theorem \ref{charsddc}, which are $d(x)d^{*}(x)\beta^{*}(x)=x^r-1$ and $a(x)a^{*}(x)\beta(x)=x^s-1$. Then we obtain a self-dual double cyclic code generated by $(d(x)\beta(x),0)$ and $(d(x),a(x))$. We see some explicit construction methods and examples for some specific lengths in the next sections.
	
	\section{Self-dual double cyclic codes of length $(r,r)$ over $\mathbb{F}_q$}
	In the previous section, we already know from Theorem \ref{existsep} that separable self-dual double cyclic codes only exist over finite fields of characteristic 2, and they can be constructed simply by taking the direct product of two self-dual cyclic codes. Therefore, our discussion shall be on the nonseparable case.
	
	This section focuses on double cyclic codes of length $(r,r)$. Note that for $r=s$, the orthogonality between $(u(x),u'(x))$ and $(v(x),v'(x))$ can be simplified to
	\begin{align*}
		(u(x),u'(x))\circ (v(x),v'(x))&
		=(u(x)v^{*}(x)+u'(x)v'^{*}(x))\Mod{x^r-1}=0.
	\end{align*}
	The next theorem provides a necessary condition for the existence of self-dual double cyclic codes of length $(r,r)$ over $\mathbb{F}_q$.
	\begin{lemma}\label{sddcc}
		If there exists a nonseparable self-dual double cyclic code of length $(r,r)$ over a finite field $\mathbb{F}_q$, then $-1$ is a square in $\mathbb{F}_q$.
	\end{lemma}
	\begin{proof}
		Let $C$ be a self-dual double cyclic code of length $(r,r)$ over $\mathbb{F}_q$ generated by $(b(x),0)$ and $(l(x),a(x))$, for some $a(x),b(x),l(x)\in\mathbb{F}_q[x]/\langle x^r-1\rangle$. Those three polynomials must satisfy $$\gcd(b(x),l(x)) b^{*}(x) = \frac{a(x)a^{*}(x)b(x)}{\gcd(b(x),l(x))}=x^{r}-1.$$
		Let $\gcd(b(x),l(x))=d(x)$. Then $b(x)=d(x)\beta(x)$ and $l(x)=d(x)\lambda(x)$, for some $\beta(x),\lambda(x)\in\mathbb{F}_q[x]$ which are relatively prime. The first condition becomes $$d(x) d^{*}(x)\beta^{*}(x) = x^r-1,$$ which implies $$d(x)d^{*}(x)=g(x)\beta(x)g^{*}(x)\beta^{*}(x)=\beta(x)(x^r-1).$$
		It is clear that $b(x)b^{*}(x)$ is self-reciprocal, and by taking the reciprocal of both sides, we get $$\beta^{*}(x)(-x^r+1)=\beta(x)(x^r-1),$$ which implies $\beta^{*}(x)=-\beta(x)$. Moreover, the second condition becomes $$a(x)a^{*}(x)\beta(x)=x^r-1,$$ and by combining it with the first condition, we get $$d(x)d^{*}(x)(-\beta(x))=a(x)a^{*}(x)\beta(x)\Leftrightarrow d(x) d^{*}(x)=-a(x)a^{*}(x).$$
		
		Let $k$ be the multiplicity of $x-1$ in the factorization of $d(x)$. Since $d(x) d^{*}(x)=-a(x)a^{*}(x)$, the multiplicity of $x-1$ in the factorization of $a(x)$ is also equal to $k$. We can write $d(x)=(x-1)^k \delta(x)$, $d^{*}(x)=(-x+1)^k\delta^{*}(x)$, $a(x)=(x-1)^k\alpha(x)$, and $a^{*}(x)=(-x+1)^k\alpha^{*}(x)$, where $\delta(1)$, $\delta^{*}(1)$, $\alpha(1)$, and $\alpha^{*}(1)$ are all nonzero. Thus, we have $\delta(x)\delta^{*}(x)=-\alpha(x)\alpha^{*}(x)$. Note that for any $f(x)\in\mathbb{F}_q[x]$, $f(1)=f^{*}(1)$. Therefore, from $\delta(1)\delta^{*}(1)=-\alpha(1)\alpha^{*}(1)$, we obtain $\left(\frac{\delta(1)}{\alpha(1)}\right)^2=-1$.
		This means that $-1$ is a square in $\mathbb{F}_q$.
	\end{proof}
	The element $-1$ is a square in $\mathbb{F}_q$ if and only if $q$ is a power of 2 or $q\equiv 1\Mod{4}$. From Lemma \ref{sddcc}, we conclude that self-dual double cyclic codes of length $(r,r)$ do not exist over $\mathbb{F}_{p^i}$, where $p\equiv 3\Mod{4}$ and $i$ is odd.
	
	\begin{remark}
		Let $C$ be a double cyclic code of length $(r,r)$ over $\mathbb{F}_q$ and $$(u|u')=(u_0,u_1,\cdots,u_{r-1}|u'_0,u'_1,\cdots,u'_{r-1})$$ be a codeword in $C$. If we rearrange the digits of $(u|u')$ to $$\bar{u}=(u_0,u'_0,u_1,u'_1,\cdots,u_{r-1},u'_{r-1}),$$ then we can see that the $(r,r)$-double cyclic shift of $(u|u')$ is equivalent to the $2$-quasi-cyclic shift of $\bar{u}$. Thus, any double cyclic code of length $(r,r)$ is equivalent to a $2$-quasi-cyclic code of length $2r$. This fact was also observed in \cite{Abualrub2024}. Moreover, Lemma \ref{sddcc} is consistent with Proposition 6.1 in \cite{Ling2001}.
	\end{remark}
	
	Now we examine the necessary and sufficient conditions for the polynomials $a(x)$, $b(x)$, and $l(x)$ (over $\mathbb{F}_q$) to become generating elements for a self-dual double cyclic code of length $(r,r)$ over $\mathbb{F}_q$.
	\begin{lemma}\label{resgcd}
		Let $\mathbb{F}_q$ be a finite field such that $-1$ is a square in $\mathbb{F}_q$. There exist polynomials $a(x),b(x),l(x)\in\mathbb{F}_q[x]$ such that $$\gcd(b(x),l(x)) b^{*}(x) = \frac{a(x)a^{*}(x)b(x)}{\gcd(b(x),l(x))}=x^r-1$$ if and only if there exists a nonconstant polynomial $f(x)\in\mathbb{F}_q[x]$ such that $f(x)f^{*}(x)$ is a factor of $x^r-1$ over $\mathbb{F}_q$.
	\end{lemma}
	\begin{proof}
		($\Rightarrow$) Let $a(x),b(x),l(x)\in\mathbb{F}_q[x]$ such that $$\gcd(b(x),l(x)) b^{*}(x) = \frac{a(x)a^{*}(x)b(x)}{\gcd(b(x),l(x))}=x^r-1.$$ If $\gcd(b(x),l(x))=d(x)$, then $b(x)=d(x)\beta(x)$, for some $\beta(x)\in\mathbb{F}_q[x]$. Moreover, $x^r-1=d(x)b^{*}(x)=d(x)d^{*}(x) \beta^{*}(x)=a(x)a^{*}(x)\beta(x)$. Here, it is clear that $d(x)$ and $a(x)$ are nonconstant polynomials such that $d(x)d^{*}(x)$ and $a(x)a^{*}(x)$ are factors of $1+x+\cdots+x^{r-1}$.
		
		($\Leftarrow$) Let $\epsilon\in\mathbb{F}_q$ such that $\epsilon^2=-1$. Suppose that there exists a nonconstant polynomial $f(x)\in\mathbb{F}_q[x]$ such that $f(x)f^{*}(x)$ is a factor of $x^r-1$ over $\mathbb{F}_q$. Without loss of generality, we can assume that $f(x)$ is monic. Then we can write $$x^r-1=f(x)f^{*}(x)p(x),$$ for some $p(x)\in\mathbb{F}_q[x]$. Taking the reciprocal of both sides, we obtain $-x^r+1=f(x)f^{*}(x)p^{*}(x)$, which implies $p^{*}(x)=-p(x)$. Choose $b(x)=- f(x)p(x)$, $l(x)=f(x)$, and $a(x)=\epsilon f(x)$. Here, $\gcd(b(x),l(x))= f(x)$, and we obtain
		$$\gcd(b(x),l(x))b^{*}(x)=[f(x)][-f^{*}(x)p^{*}(x)]=f(x)f^{*}(x)p(x)=x^r-1$$ and $$\frac{a(x)a^{*}(x)b(x)}{\gcd(b(x),l(x))}=\frac{[\epsilon f(x)][\epsilon f^{*}(x)][- f(x)p(x)]}{ f(x)}=-\epsilon^2 f(x)f^{*}(x)p(x)=x^r-1.$$
		Moreover, it can also be shown that $\{(b(x),0),(l(x),a(x))\}$ is an orthogonal set.
	\end{proof}
	
	By combining Lemma \ref{sddcc} and Lemma \ref{resgcd}, the existence of nonseparable self-dual double cyclic codes of length $(r,r)$ over $\mathbb{F}_q$ can be characterized in the following theorem.
	\begin{theorem}\label{kareks}
		There exists a self-dual double cyclic code of length $(r,r)$ over $\mathbb{F}_q$ if and only if these two conditions are satisfied:
		\begin{itemize}
			\item $-1$ is a square in $\mathbb{F}_q$, i.e. there exists $\epsilon\in\mathbb{F}_q$ such that $\epsilon^2=-1$, and
			\item there exists a nonconstant polynomial $f(x)\in\mathbb{F}_q[x]$ such that $f(x)f^{*}(x)$ is a factor of $x^r-1$ over $\mathbb{F}_q$.
		\end{itemize}
	\end{theorem}
	Theorem \ref{kareks} already provides the necessary and sufficient conditions for the existence of self-dual double cyclic codes of length $(r,r)$ over $\mathbb{F}_q$, which depend on the factorization of $x^r-1$ over $\mathbb{F}_q$. If the characteristic of $\mathbb{F}_q$ is $p$, then $q=p^m$, for some $m\in\mathbb{N}$.
	Let $r$ be a multiple of $p$, i.e. $r=\alpha p$, for some $\alpha\in\mathbb{N}$. We have $$x^r-1=x^{\alpha p}-1=(x^\alpha-1)^p=(x-1)^p (1+x+\cdots+x^{\alpha-1})^p.$$
	Here, $f(x)=x-1$ satisfies the second condition of Theorem \ref{kareks} since $f(x)f^{*}(x)=-(x-1)^2$ divides $x^r-1$. Therefore, in this case, such a self-dual double cyclic code exists.
	\begin{corollary}\label{teo51}
		Let $\mathbb{F}_{p^m}$ be a finite field such that $-1$ is a square in $\mathbb{F}_{p^m}$. For any $\alpha\in\mathbb{N}$, there exists a nonseparable self-dual double cyclic code of length $(\alpha p,\alpha p)$ over $\mathbb{F}_{p^m}$.
	\end{corollary}
	
	Self-dual double cyclic codes of length $(r,r)$ over $\mathbb{F}_q$ can be constructed as in the proof of Lemma \ref{resgcd}. If $f(x)\in\mathbb{F}_q[x]$ is a nonconstant polynomial such that $f(x)f^{*}(x)$ is a factor of $x^r-1$, then we can construct a self-dual double cyclic code of length $(r,r)$ over $\mathbb{F}_q$ generated by $(b(x),0)$ and $(l(x),a(x))$, where $b(x)=-\frac{x^r-1}{f^{*}(x)}$, $l(x)=f(x)$, and $a(x)=\epsilon f(x)$.
	\begin{remark}
		It is a well-known fact that a factorization of $x^r-1$ over $\mathbb{F}_{p^m}$ has repeated irreducible factors if and only if $p$ divides $r$. Thus, in the case where $r$ is not a multiple of $p$, in the factorization of $x^r-1$ over $\mathbb{F}_{p^m}$, each of its irreducible factors appears only once.
		
		If $f(x)$ is an irreducible factor of $x^r-1$ other than $x-1$, then both $f(x)$ and $f^{*}(x)$ are factors of $1+x+\cdots+x^{r-1}$. Consequently, $f(x)f^{*}(x)$ is also a factor of $x^r-1$, unless $f(x)$ is self-reciprocal. In other words, there does not exist a self-dual double cyclic code of length $(r,r)$ over $\mathbb{F}_{p^m}$, where $r$ is not a multiple of $p$, if $1+x+\cdots+x^{r-1}$ only has self-reciprocal factors.
		
		As examples, there does not exist a self-dual double cyclic code of length $(3,3)$ over $\mathbb{F}_2$ since $1+x+x^2$ is irreducible over $\mathbb{F}_2$, and there does not exist a self-dual double cyclic code of length $(6,6)$ over $\mathbb{F}_5$ since $1+x+x^2+x^3+x^4+x^5=(1+x)(1+x+x^2)(1+4x+x^2)$.
	\end{remark}
	We conclude this section with some explicit examples of self-dual double cyclic codes of length $(r,r)$ over some specific finite fields $\mathbb{F}_q$. We will see some examples of such codes over $\mathbb{F}_2$, $\mathbb{F}_4$, and $\mathbb{F}_5$. In $\mathbb{F}_2$ and $\mathbb{F}_4$, $-1$ is the square of 1, while in $\mathbb{F}_5$, $-1$ is the square of 2 and 3. All minimum distances of the codes obtained in the following examples are determined using the free MAGMA calculator (\url{https://magma.maths.usyd.edu.au/calc/}). 
	\begin{example}\label{ex27}
		Here $q=2$ and $r=7$. Since $x^7-1=(1+x)(1+x+x^3)(1+x^2+x^3)$ in $\mathbb{F}_2[x]$, both $1+x+x^3$ and $1+x^2+x^3$ satisfy the second condition of Theorem \ref{kareks}. For $f(x)=1+x+x^3$, we can choose $a(x)=l(x)=f(x)=1+x+x^3$ and $b(x)=\frac{x^7-1}{f^{*}(x)}=\frac{x^7-1}{1+x^2+x^3}=1+x^2+x^3+x^4$. Then we obtain a self-dual double cyclic code of length $(7,7)$ over $\mathbb{F}_2$ generated by $(b(x),0)=(1+x^2+x^3+x^4,0)$ and $(l(x),a(x))=(1+x+x^3,1+x+x^3)$. As a binary linear code of length 14, it is generated by
		\begin{align*}
			\{&1011100\; 0000000,0101110\; 0000000,0010111\; 0000000,1101000\; 1101000,\\
			&0110100\; 0110100,0011010\; 0011010,0001101\; 0001101\}.
		\end{align*}
		In a similar way, for $f(x)=1+x^2+x^3$, we obtain another self-dual double cyclic code of length $(7,7)$ over $\mathbb{F}_2$ generated by $(1+x+x^2+x^4,0)$ and $(1+x^2+x^3,1+x^2+x^3)$. As a binary linear code of length 14, it is generated by
		\begin{align*}
			\{&1110100\; 0000000,0111010\; 0000000,0011101\; 0000000,1011000\; 1011000,\\
			&0101100\;0101100,0010110\;0010110,0001011\;0001011\}.
		\end{align*}
		Moreover, both codes have minimum distance 4. From \cite{Grassl}, the largest possible minimum distance of binary $[14,7]$-codes is 4. Therefore, the two codes obtained in this example are optimal codes.
	\end{example}
	
	\begin{example}\label{exampleF4}
		Here $q=4$ and $r=3$. In $\mathbb{F}_4[x]$, $x^3-1=(1+x)(a+x)(a^2+x)$. Since the reciprocal of $f_1(x)=a+x$ is $f_1^{*}(x)=1+ax=a(a^2+x)$, this $f_1(x)$ satisfies the second condition of Theorem \ref{kareks}. We can choose $a(x)=l(x)=f_1(x)=a+x$ and $b(x)=\frac{x^3-1}{f_1^{*}(x)}=\frac{x^3-1}{1+ax}=a^2(1+x)(a+x)=1+ax+a^2 x^2$. Then we obtain a self-dual double cyclic code of length $(3,3)$ over $\mathbb{F}_4$ generated by $(b(x),0)=(1+a x+a^2 x^2,0)$ and $(l(x),a(x))=(a+x,a+x)$. As a linear code of length 6 over $\mathbb{F}_4$, it is generated by
		$$\{1aa^2\; 000,a10\; a10,0a1\; 0a1\}.$$
		Furthermore, the polynomial $f_2(x)=a^2+x$ also satisfies the second condition of Theorem \ref{kareks}. In a similar way, we obtain another self-dual double cyclic code of length $(3,3)$ over $\mathbb{F}_4$ generated by $(1+a^2x+a x^2,0)$ and $(a^2+x,a^2+x)$. As a linear code of length 6 over $\mathbb{F}_4$, it is generated by
		$$\{1a^2a\; 000,a^210\; a^210,0a^21\; 0a^21\}.$$
		Moreover, both codes have minimum distance 3. Based on \cite{Harada}, these codes are optimal self-dual codes.
	\end{example}
	
	\begin{example}\label{exampleF5b}
		Here $q=5$ and $r=5$. In $\mathbb{F}_5[x]$, $x^5-1=(4+x)^5.$ There are two factors of $x^5-1$ that satisfy the second condition of Theorem \ref{kareks}, which are $f_1(x)=4+x$ and $f_2(x)=(4+x)^2=1+3x+x^2$. Note that there are two $\epsilon\in\mathbb{F}_5$ such that $\epsilon^2=-1$, which are $\epsilon=2$ and $\epsilon=3$. Thus, both are considered when constructing the code.
		
		For $f_1(x)=4+x$, if $\epsilon=2$, then we can choose $b(x)=-\frac{x^5-1}{f_1^{*}(x)}=4\frac{x^5-1}{1+4x}=1+x+x^2+x^3+x^4$, $l(x)= f_1(x)=4+x$, and $a(x)=\epsilon f_1(x)=2(4+x)=3+2x$. Here we obtain a self-dual double cyclic code of length $(5,5)$ over $\mathbb{F}_5$ generated by $\{(1+x+x^2+x^3+x^4,0),(4+x,3+2x)\}$. If $\epsilon=3$, then we obtain a self-dual double cyclic code of length $(5,5)$ over $\mathbb{F}_5$ generated by $\{(1+x+x^2+x^3+x^4,0),(4+x,2+3x)\}$.
		
		In a similar way, for $f_2(x)=1+3x+x^2$, we obtain two self-dual double cyclic codes of length $(5,5)$ over $\mathbb{F}_5$ generated by $\{(1+2x+3x^2+4x^3,0),(1+3x+x^2,2+x+2x^2)\}$ and $\{(1+2x+3x^2+4x^3,0),(1+3x+1x^2,3+4x+3x^2)\}$.
		
		As vectors of length 10 over $\mathbb{F}_5$, the corresponding generating sets are
		\begin{itemize}
			\item $\{11111\;00000,41000\;32000,04100\;03200,00410\;00320,00041\;00032\},$
			\item $\{11111\;00000,41000\;23000,04100\;02300,00410\;00230,00041\;00023\},$
			\item $\{12340\;00000,01234\;00000,13100\;21200,01310\;02120,00131\;00212\},$
			\item $\{12340\;00000,01234\;00000,13100\;34300,01310\;03430,00131\;00343\}.$
		\end{itemize}
		It can be shown that the four above sets are orthogonal sets over $\mathbb{F}_5$.
		Moreover, all these four codes have minimum distance 4. Based on \cite{Harada}, all these codes are optimal self-dual codes.
	\end{example}
	
	\section{Self-dual double cyclic codes of length $(r,2r)$ and $(2r,r)$ over $\mathbb{F}_{q}$}
	
	In this section, we examine the existence of self-dual double cyclic codes of length $(r,2r)$ and $(2r,r)$ over a finite field $\mathbb{F}_q$. Note that in the case of length $(r,2r)$, the orthogonality between $(u(x),u'(x))$ and $(v(x),v'(x))$ can be simplified to $$(u(x),u'(x))\circ (v(x),v'(x))=\left[u(x)v^{*}(x)\Mod{x^r-1}\right](x^r+1)+u'(x)v'^{*}(x)\Mod{x^{2r}-1}=0,$$
	and in the case of length $(2r,r)$, the orthogonality between $(u(x),u'(x))$ and $(v(x),v'(x))$ can be simplified to $$(u(x),u'(x))\circ (v(x),v'(x))=u(x)v^{*}(x)\Mod{x^{2r}-1}+\left[u'(x)v'^{*}(x)\Mod{x^r-1}\right](x^r+1)=0.$$
	We start with the special case of codes over a finite field of characteristic 2. We provide a characterization for the existence, followed by explicit examples.
	\begin{lemma}\label{case1}
		Let $\mathbb{F}_q$ be a finite field of characteristic 2. There exist 
		self-dual double cyclic codes of length $(r,2r)$ and $(2r,r)$ over $\mathbb{F}_q$ if and only if $r$ is even.
	\end{lemma}
	\begin{proof}
		In both cases, the total length is $3r$, and since self-dual codes are always of even length, $r$ is necessarily even. Conversely, if $r$ is even, then in a finite field of characteristic 2 we have $x^r-1=x^r+1=(x^{r/2}+1)^2$. Since $x+1$ divides $x^{r/2}+1$, $(x+1)^2$ divides $x^r+1$. To be precise, $x^r+1=(x^{r/2}+1)^2=(1+x)^2(1+x+\cdots+x^{r/2-1})^2$.
		
		For the case of length $(r,2r)$, choose $a(x)=(1+x)(1+x^{r/2})$, $b(x)=1+x+\cdots+x^{r-1}=(1+x)(1+x+\cdots+x^{r/2-1})^2$, and $l(x)=1+x$. It can be verified that $\{(b(x),0),(l(x),a(x))\}$ is an orthogonal set over $\mathbb{F}_q$. Therefore, the code generated by $(b(x),0)$ and $(l(x),a(x))$ above is a self-dual double cyclic code of length $(r,2r)$ over $\mathbb{F}_q$.
		
		For the case of length $(2r,r)$, choose $a(x)=1+x$, $b(x)=(1+x+\cdots+x^{r-1})(1+x^{r/2})$, and $l(x)=(1+x)(1+x^{r/2})$. It can be shown in a similar way that the code generated by $(b(x),0)$ and $(l(x),a(x))$ is a self-dual double cyclic code of length $(2r,r)$ over $\mathbb{F}_q$.
	\end{proof}
	\begin{example}\label{ex52}
		Here $q=2$ and $r=4$. The self-dual double cyclic codes of length $(4,8)$ and $(8,4)$ over $\mathbb{F}_2$ described in the proof of Lemma \ref{case1} are generated by $\{(1+x+x^2+x^3,0),(1+x,1+x+x^2+x^3)\}$ and $\{(1+x+x^4+x^5,0),(1+x+x^2+x^3,1+x)\}$, respectively. Moreover, both codes have minimum distance 4. From \cite{Grassl}, the largest possible minimum distance of binary $[12,6]$-codes is 4. Therefore, the two codes obtained in this example are optimal codes.
	\end{example}
	Now we discuss the case of codes over a finite field of odd (prime) characteristic. It turns out that this case is also related to self-dual negacyclic codes. Therefore, we need to recall some basic facts about negacyclic codes. Structurally, a negacyclic code of length $n$ over $\mathbb{F}_{q}$ is an ideal of $\mathbb{F}_{q}[x]/\langle x^n+1\rangle$ \cite{Guenda2012,Jitman2020}. Moreover, similar to classical cyclic codes, a negacyclic code of length $n$ over $\mathbb{F}_{q}$ is generated by a polynomial which is a factor of $x^n+1$.
	\begin{lemma}\label{sddcc2}
		Let $\mathbb{F}_q$ be a finite field of odd characteristic. If there exists a nonseparable self-dual double cyclic code of length $(r,2r)$ and $(2r,r)$ over $\mathbb{F}_q$, then $-2$ is a square in $\mathbb{F}_q$.
	\end{lemma}
	\begin{proof}
		We prove for the case of length $(r,2r)$, since the other case of length $(2r,r)$ can be proved in a similar way. From Theorem \ref{charsddc}, to construct a self-dual double cyclic code of length $(r,2r)$, we need to find polynomials $a(x)$, $b(x)$, and $l(x)$ over $\mathbb{F}_q[x]$ such that $$\gcd(b(x),l(x)) b^{*}(x) = x^r-1$$
		and $$\frac{a(x)a^{*}(x)b(x)}{\gcd(b(x),l(x))}=x^{2r}-1.$$
		Let $\gcd(b(x),l(x))=d(x)$ and $b(x)=d(x)\beta(x)$, for some $\beta(x)\in\mathbb{F}_q[x]$. As in the proof of Lemma \ref{sddcc}, here $\beta^{*}(x)=-\beta(x)$, and hence $$x^r+1=\frac{x^{2r}-1}{x^r-1}=\frac{a(x)a^{*}(x)\beta(x)}{d(x)d^{*}(x)\beta^{*}(x)}=-\frac{a(x)a^{*}(x)}{d(x)d^{*}(x)}.$$
		Similar to the proof of Lemma \ref{sddcc}, by cancelling any factors of $x-1$, we have $x^r+1=-\frac{\alpha(x)\alpha^{*}(x)}{\delta(x)\delta^{*}(x)}$. Taking $x=1$, we obtain $-2=\left(\frac{\alpha(1)}{\delta(1)}\right)^2$, i.e. $-2$ is a square in $\mathbb{F}_q$.
	\end{proof}
	\begin{remark}
		Using some facts in number theory, for finite fields of odd characteristic, $-2$ is a square in $\mathbb{F}_q$ if and only if $q\equiv 1\Mod{8}$ or $q\equiv 3\Mod{8}$.
	\end{remark}
	Based on these factorizations, $d(x)$ and $a(x)$ in the proof of Lemma \ref{sddcc2} can be chosen such that $d(x)$ divides $a(x)$ (for the case of length $(2r,r)$, $a(x)$ divides $d(x)$). Then $x^r+1=-e(x)e^{*}(x)$, where $e(x)=\frac{a(x)}{d(x)}$ (for the case of length $(2r,r)$, $e(x)=\frac{d(x)}{a(x)}$). Here, $e(x)$ is actually a generator polynomial of a self-dual negacyclic code. If we assume that $e(x)$ is a monic polynomial, then the constant term of $e(x)$ must be $-1$. It is important to note that not every self-dual negacyclic code has a monic generator polynomial with constant term $-1$. The following lemma provides the necessary and sufficient conditions for the existence of such $e(x)$.
	\begin{lemma}\label{charnega}
		Let $\mathbb{F}_q$ be a finite field of odd characteristic. There exists a polynomial $e(x)\in\mathbb{F}_q[x]$ such that $-e(x)e^{*}(x)=x^r+1$ if and only if $-2$ is a square in $\mathbb{F}_q$ and $r$ is divisible by 4.
	\end{lemma}
	\begin{proof}
		($\Leftarrow$) Let $\delta\in\mathbb{F}_q$ such that $\delta^2=-2$. For any $r$ that is divisible by 4, we can define $e(x)=x^{r/2}+\delta x^{r/4}-1$. Here, $e^{*}(x)=-x^{r/2}+\delta x^{r/4}+1$, and \begin{align*}
			-e(x)e^{*}(x)&=-(x^{r/2}+\delta x^{r/4}-1)(-x^{r/2}+\delta x^{r/4}+1)=(x^{r/2}+\delta x^{r/4}-1)(x^{r/2}-\delta x^{r/4}-1)\\
			&=x^r+(\delta-\delta)x^{3r/4}+(-\delta^2-2)x^{r/2}+(\delta-\delta)x^{r/4}+1=x^r+1.
		\end{align*}
		($\Rightarrow$) Since it is clear that $r$ must be even, we only need to consider the case where $r$ is even, but not divisible by 4. Using a similar idea as in \cite{Boripan2021}, let $\gamma$ be the primitive $2r$-th root of unity in some extension field of $\mathbb{F}_q$ (i.e. $\gamma^{2r}=1$ and $\gamma^r=-1$). Since $x^{2r}-1=(x^r-1)(x^r+1)$, $x^{2r}-1=(x-\gamma)(x-\gamma^2)\cdots(x-\gamma^{2r}),$ and $x^r-1=(x-\gamma^2)(x-\gamma^4)\cdots(x-\gamma^{2r}),$ we obtain $$x^r+1=(x-\gamma)(x-\gamma^3)\cdots (x-\gamma^{2r-1}).$$ Note that $(x-\gamma)^{*}=-\gamma x+1=-\gamma(x-\gamma^{2r-1})$, $(x-\gamma^3)^{*}=-\gamma^3 x+1=-\gamma^3(x-\gamma^{2r-3})$, $\cdots$, and $(x-\gamma^{r-1})^{*}=-\gamma^{r-1} x+1=-\gamma^{r-1}(x-\gamma^{r+1})$.
		
		Suppose that $e(x)$ is a polynomial over $\mathbb{F}_q$ such that $e(x)e^{*}(x)=\kappa (x^r+1)$, for some $\kappa\in\mathbb{F}_q$. Without loss of generality, if we assume that $e(x)$ is monic, then $\kappa$ is the constant term of $e(x)$ and $e(x)=e_1(x)e_2(x)\cdots e_{r/2}(x)$, where $e_i(x)$ is one of $x-\gamma^{2i-1}$ or $x-\gamma^{2(r-i)+1}$, for $i=1,2,\cdots,r/2$. Thus, the constant term of $e(x)$ is $\kappa=(-\gamma^{e_1})(-\gamma^{e_2})\cdots(-\gamma^{e_{r/2}})$, for some odd positive integers $e_1,e_2,\cdots,e_{r/2}$. Since $r$ is even but not divisible by 4, $r/2$ is odd, and $\kappa=-\gamma^{e_1+e_2+\cdots+e_{r/2}}$. Here, $e_1+e_2+\cdots+e_{r/2}$ is odd and therefore $\kappa$ cannot be equal to $-1$. Hence, there does not exist a polynomial $e(x)\in\mathbb{F}_q[x]$ such that $-e(x)e^{*}(x)=x^r+1$.
	\end{proof}
	Note that $e(x)=x^{r/2}-\delta x^{r/4}-1$ also satisfies $-e(x)e^{*}(x)=x^r+1$. Thus, for any $r$ that is divisible by 4, we always have $e(x)=x^{r/2}\pm \delta x^{r/4}-1$ as solutions of $-e(x)e^{*}(x)=x^r+1$ in $\mathbb{F}_q[x]$, provided that $\delta^2=-2$ in $\mathbb{F}_q$. Furthermore, in the following theorem we provide a characterization for the existence of self-dual double cyclic codes of length $(r,2r)$ and $(2r,r)$.
	\begin{theorem}\label{caser2r}
		Let $\mathbb{F}_q$ be a finite field of odd characteristic. There exists a self-dual double cyclic code of length $(r,2r)$ and $(2r,r)$ over $\mathbb{F}_q$ if and only if these three conditions are satisfied:
		\begin{itemize}
			\item $-2$ is a square in $\mathbb{F}_q$, i.e. there exists $\delta\in\mathbb{F}_q$ such that $\delta^2=-2$,
			\item $r$ is divisible by 4, and
			\item there exists a nonconstant polynomial $f(x)\in\mathbb{F}_q[x]$ such that $f(x)f^{*}(x)$ is a factor of $x^r-1$ over $\mathbb{F}_q$.
		\end{itemize}
	\end{theorem}
	\begin{proof}
		($\Rightarrow$) As in the proof of Lemma \ref{sddcc2}, the existence of self-dual double cyclic codes of length $(r,2r)$ and $(2r,r)$ implies that $x^r+1=-e(x)e^{*}(x)$, for some $e(x)\in\mathbb{F}_q[x]$. From Lemma \ref{charnega}, this is equivalent to the first two conditions. The third condition is clearly satisfied since in the case of length $(r,2r)$, we have $d(x)d^{*}(x)\beta^{*}(x)=x^r-1$, where $d(x)=\gcd(b(x),l(x))$, and in the case of length $(2r,r)$, we have $a(x)a^{*}(x)=x^r-1$.
		
		($\Leftarrow$) For the case of length $(r,2r)$, if the first two conditions are satisfied, then as described in the proof of Lemma \ref{charnega}, there exists $e(x)=x^{r/2}+\delta x^{r/4}-1$, where $\delta\in\mathbb{F}_q$ such that $\delta^2=-2$. Moreover, if the third condition is satisfied, we can choose $a(x)=f(x)e(x)$, $b(x)=\frac{x^r-1}{f^{*}(x)}$, and $l(x)=f(x)$. It can be verified that $\{(b(x),0),(l(x),a(x))\}$ is an orthogonal set. Therefore, the code generated by $(b(x),0)$ and $(l(x),a(x))$ is a self-dual double cyclic code of length $(r,2r)$ over $\mathbb{F}_q$.
		
		The proof is similar for the case of length $(2r,r)$ by choosing $a(x)=f(x)$, $b(x)=\frac{x^r-1}{f^{*}(x)}e(x)$, and $l(x)=f(x)e(x)$. By verifying that $\{(b(x),0),(l(x),a(x))\}$ is an orthogonal set, the code generated by that set is a self-dual double cyclic code of length $(2r,r)$ over $\mathbb{F}_q$.
	\end{proof}
	In other words, we can construct self-dual double cyclic codes of length $(r,2r)$ and $(2r,r)$ by choosing generating elements $(b(x),0)$ and $(l(x),a(x))$ as described in the proof of Theorem \ref{caser2r}. We provide an example for each case over $\mathbb{F}_3$.
	\begin{example}\label{exr2r}
		Here $q=3$ and $r=8$. We construct a self-dual double cyclic code of length $(r,2r)=(8,16)$ and $(2r,r)=(16,8)$ over $\mathbb{F}_3$. In $\mathbb{F}_3[x]$, $x^8-1=(1+x)(2+x)(1+x^2)(2+x+x^2)(2+2x+x^2)$ and $x^8+1=(2+x^2+x^4)(2+2x^2+x^4)$. Choose $f(x)=2+x+x^2$ and $e(x)=2+x^2+x^4$.
		For the length $(8,16)$,
		\begin{itemize}
			\item $b(x)=\frac{x^8-1}{f^{*}(x)}=2+x+x^2+x^4+2x^5+2x^6$,
			\item $l(x)=f(x)=2+x+x^2$, and
			\item $a(x)=f(x)e(x)=1+2x+x^2+x^3+x^5+x^6$.
		\end{itemize}
		Then we obtain a self-dual double cyclic code of length $(8,16)$ over $\mathbb{F}_3$ generated by $$\{(2+x+x^2+x^4+2x^5+2x^6,0),(2+x+x^2,1+2x+x^2+x^3+x^5+x^6)\}.$$
		For the length $(16,8)$,
		\begin{itemize}
			\item $b(x)=\frac{x^8-1}{f^{*}(x)}e(x)=1+2x+x^2+x^3+2x^4+2x^5+2x^7+2x^9+2x^{10}$,
			\item $l(x)=f(x)e(x)=1+2x+x^2+x^3+x^5+x^6$, and
			\item $a(x)=f(x)=2+x+x^2$.
		\end{itemize}
		Then we obtain a self-dual double cyclic code of length $(16,8)$ over $\mathbb{F}_3$ generated by $$\{(1+2x+x^2+x^3+2x^4+2x^5+2x^7+2x^9+2x^{10},0),(1+2x+x^2+x^3+x^5+x^6,2+x+x^2)\}.$$
		Moreover, both codes have minimum distance 6.
	\end{example}
	Note that in Example \ref{exr2r}, we can obtain other self-dual codes of length 24 over $\mathbb{F}_3$ by choosing different $f(x)$ and $e(x)$, for example $f(x)=2+2x+x^2$ and $e(x)=2+2x^2+x^4$. Again, the existence depends on the factorization of $x^r-1$ and $x^r+1$ over $\mathbb{F}_q$.
	
	\section{Self-dual double cyclic codes of length $(r,s)$ over $\mathbb{F}_{q}$, where $\gcd(r,s)=1$}
	In this section, we examine the existence of self-dual double cyclic codes over $\mathbb{F}_q$ of another specific length. In the previous two sections, we examined the cases where $s$ is a multiple of $r$, or the other way around. We now consider the case where $r$ and $s$ are relatively prime, i.e. $\gcd(r,s)=1$.
	
	Recall that from Theorem \ref{charsddc}, we need to find polynomials $d(x)$, $\beta(x)$, and $a(x)$ over $\mathbb{F}_q$ such that $d(x)d^{*}(x)\beta^{*}(x)=x^r-1$ and $a(x)a^{*}(x)\beta(x)=x^s-1$. Since $\beta^{*}(x)=-\beta(x)$, $\beta(x)$ is a common divisor of $x^r-1$ and $x^s-1$. Hence, $\beta(x)$ divides $\gcd(x^r-1,x^s-1)=x^{\gcd(r,s)}-1$. In the case where $\gcd(r,s)=1$, $\beta(x)$ divides $x-1$. If $\beta(x)$ is a nonzero constant polynomial, then the obtained code is separable. To obtain a nonseparable one, the only possibility is $\beta(x)=\alpha(x-1)$, for some nonzero $\alpha\in\mathbb{F}_q$, in which those two conditions become
	$$-\alpha d(x)d^{*}(x)=1+x+\cdots+x^{r-1}$$
	and $$\alpha a(x)a^{*}(x)=1+x+\cdots+x^{s-1}.$$
	
	It is clear that in this case both $r$ and $s$ must be odd. We need to find $r$ and $s$ such that $\gcd(r,s)=1$ and the factorizations of $1+x+\cdots+x^{r-1}$ and $1+x+\cdots+x^{s-1}$ satisfy the two conditions described above. Then we can construct a self-dual double cyclic code of length $(r,s)$ over $\mathbb{F}_q$ generated by $\{(d(x)(x-1),0),(d(x),a(x))\}$. We provide some explicit examples over various finite fields $\mathbb{F}_q$.
	
	\begin{example}\label{ex3coprime}
		For $q=3$, we can work on the case $r=9$ and $s=11$, because
		\begin{align*}
			1+x+\cdots+x^8&=(1+2x+2x^3+x^4)(1+2x+2x^3+x^4),\\
			1+x+\cdots+x^{10}&=(2+2x+x^2+2x^3+x^5) (2+x^2+2x^3+x^4+x^5).
		\end{align*}
		Here $\alpha=-1=2$, and we can choose $d(x)=1+2x+2x^3+x^4$ and $a(x)=2+2x+x^2+2x^3+x^5$. Then we obtain a self-dual double cyclic code generated by $$(d(x)2(x-1),0)=(1+x+x^2+2x^3+2x^4+2x^5,0)$$ and $$(d(x),a(x))=(1+2x+2x^3+x^4,2+2x+x^2+2x^3+x^5).$$
		Note that $a(x)=2+x^2+2x^3+x^4+x^5$ also satisfies the condition. In this case, we obtain another self-dual double cyclic code, which is
		$$\langle (1+x+x^2+2x^3+2x^4+2x^5,0),(1+2x+2x^3+x^4,2+x^2+2x^3+x^4+x^5)\rangle.$$
		Moreover, both codes have minimum distance 3.
	\end{example}
	
	\begin{example}\label{ex4coprime}
		For $q=4$, we can work on the case $r=7$ and $s=9$, because
		\begin{align*}
			1+x+\cdots+x^6&=(1+x+x^3)(1+x^2+x^3),\\
			1+x+\cdots+x^8&=(a+x)(a^2+x)(a+x^3)(a^2+x^3)\\
			&=(1+a x+a^2 x^3+x^4)(1+a^2 x+a x^3+x^4).
		\end{align*}
		Here $\alpha=-1=1$, and we can choose $d(x)=1+x+x^3$ and $a(x)=1+a x+a^2 x^3+x^4$. Then we obtain a self-dual double cyclic code generated by $(1+x^2+x^3+x^4,0)$ and $(1+x+x^3,1+a x+a^2 x^3+x^4).$ The minimum distance of this code is 3.
	\end{example}
	
	\begin{example}\label{ex7coprime}
		For $q=7$, we can work on the case $r=7$ and $s=9$, because
		\begin{align*}
			1+x+\cdots+x^6&=(6+x)^6=6(6+x)^3(1+6x)^3\\
			&=6(6+3x+4x^2+x^3)(1+4x+3x^2+6x^3),\\
			1+x+\cdots+x^8&=(3+x)(5+x)(3+x^3)(5+x^3)\\
			&=(3+x)(1+3x)(5+x^3) (1+5x^3) \\
			&= (1+3x+5x^3+x^4)(1+5x+3x^3+x^4).
		\end{align*}
		Here $\alpha=1$, and we can choose $d(x)=6+3x+4x^2+x^3$ and $a(x)=1+3x+5x^3+x^4$. Then we obtain a self-dual double cyclic code generated by $(1+3x+6x^2+3x^3+x^4,0)$ and $(6+3x+4x^2+x^3,1+3x+5x^3+x^4).$ The minimum distance of this code is 3.
	\end{example}
	Note that in all the examples provided in this section, by swapping the partition, we can also construct double cyclic codes of length $(s,r)$ over $\mathbb{F}_q$. The existence of polynomials $d(x)$ and $a(x)$ such that $-\alpha d(x)d^{*}(x)=1+x+\cdots+x^{r-1}$ and $\alpha a(x)a^{*}(x)=1+x+\cdots+x^{s-1}$, for some $\alpha\in\mathbb{F}_q$, depends on the factorization of $1+x+\cdots+x^{r-1}$ and $1+x+\cdots+x^{s-1}$ over $\mathbb{F}_q$. Not every pair of relatively prime odd numbers $(r,s)$ satisfies these conditions.
	
	\section{Concluding remark}
	In this article, we have investigated some properties of self-dual double cyclic codes over a finite field $\mathbb{F}_q$. 
	Among our main contributions are the necessary and sufficient conditions for the polynomials in the generating set of self-dual double cyclic codes. This characterization works for any length $(r,s)$ and over any finite field $\mathbb{F}_q$.
	
	We also observe some special cases of self-dual double cyclic codes of a specific length. In each case, we characterize the existence of such codes, which mainly depends on the polynomial factorization over finite fields. We also find that in some cases, self-dual double cyclic codes only exist over some specific finite fields. Based on the characterizations, we formalize an explicit construction method and obtain some explicit examples. Moreover, some of the examples are optimal self-dual codes.
	
	We also find that some cases are related to other classes of codes. Self-dual double cyclic codes of length $(r,r)$ are equivalent to self-dual 2-quasi-cyclic codes, while self-dual double cyclic codes of length $(r,2r)$ and $(2r,r)$ are related to self-dual negacyclic codes. Therefore, it is theoretically possible to apply our results to construct other classes of codes.
	
	\section*{List of Abbreviations}
	Not applicable.
	
	\section*{Declarations}
	\subsection*{Acknowledgements}
	The authors thank Prof. Akihiro Munemasa from Tohoku University, Japan, for his useful feedback. The authors are also grateful to the reviewers for their thorough reading of the manuscript
	and their helpful comments.
	
	\subsection*{Authors' Contributions}
	R.A. performed the main technical work, including result development, examples, computations, and drafting the manuscript.
	A.B. contributed to conceptualization, theoretical development, result validation, and improvement of exposition.
	D.S. contributed to conceptual and theoretical work, manuscript revision, and primary supervision and project management.
	All authors approved the final manuscript.
	
	\subsection*{Funding}
	This research was supported by the Institut Teknologi Bandung (ITB) and the Ministry of Higher Education, Science, and Technology
	(\emph{Kementerian Pendidikan Tinggi, Sains, dan Teknologi (Kemdiktisaintek)}),
	Republic of Indonesia. The first author was also supported by the Universitas Sanata Dharma (USD) for his PhD studies at ITB. 
	
	\subsection*{Competing interest}
	The authors declare no competing interests.
	
	\subsection*{Data availability}
	No datasets were generated or analysed during the current study.
	
	\subsection*{Material availability}
	Not applicable.
	\bibliographystyle{amsplain}

\end{document}